\documentclass[11pt]{llncs}

\usepackage{amsmath,amssymb}
\usepackage{graphicx}

\allowdisplaybreaks
\def\htab{\hspace*{3em}}

\def\YES{\texttt{YES}}

\def\poly{\textrm{poly}}
\def\mathdef{\stackrel{\mathrm{def}}{=}}

\newcommand{\nop}[1]{}

\def\E{\textrm{E}}

\def\htab{\hspace*{3em}}
\def\exprun{\widetilde{O}}

\def\allbutfalse{\{0,1\}^3\setminus 0^3}

\newtheorem{fact}{Fact}

\begin{document}

\title{
%
Derandomizing HSSW Algorithm for 3-SAT
%
%
}

\author{
Kazuhisa Makino \inst{1}
\and
Suguru Tamaki\inst{2}
\and
Masaki Yamamoto\inst{3}
}

\institute{
Graduate School of Information Science and Technology,
University of Tokyo\\
\email{makino@mist.i.u-tokyo.ac.jp}
\and
Graduate School of Informatics,
Kyoto University\\
\email{tamak@kuis.kyoto-u.ac.jp}
\and
Dept.\ of Informatics,
Kwansei-Gakuin University\\
\email{masaki.yamamoto@kwansei.ac.jp}
}

\date{}

\maketitle

\begin{abstract}
We present a (full) derandomization of 
HSSW algorithm for 3-SAT, 
proposed by Hofmeister, Sch\"oning, Schuler, and Watanabe in [STACS'02].
Thereby,
we obtain an $\exprun(1.3303^n)$-time deterministic algorithm for 3-SAT,
which is currently fastest.
\end{abstract}

\section{Introduction}

The satisfiability problem (SAT) is one of the most fundamental NP-hard
problems.
Questing for faster (exponential-time) {\em exact algorithms} is one of
the main research directions on SAT.
Initiated by Monien and Speckenmeyer~\cite{MS85}, a number of algorithms
for {\em exactly} solving SAT have been proposed, and many important
techniques to analyze those algorithms have been developed~\cite{DH09}.
See also~\cite{DGH02,MS10,PPS98,PPZ97,Sch99}, for example.
The most well-studied restriction of the satisfiability problem is
3-SAT~\cite{BS03,BK04,HMS11,HSS02,ISTT10,IT04,Rol03,Rol06,Sch08},
i.e., the CNF satisfiability problem with clauses of length at most
three.
The currently best known time complexities for 3-SAT are
$\exprun(1.3211^n)$ achieved by randomized
algorithms~\cite{HMS11} and $\exprun(1.3334^n)$
derived by deterministic algorithms~\cite{MS10}, where $n$ denotes the
number of Boolean variables in the formula.

As we can see, there is a 
noticeable gap between the current randomized and deterministic time
bounds for 3-SAT.
This raises a natural question: Can we close the gap completely?
One promising way to attack the above question is
\textit{derandomization}.
Roughly speaking, the task of derandomization is to construct an
algorithm which deterministically and efficiently simulates the
original randomized algorithm. 
There are a lot of strong derandomization results,
e.g.~\cite{AKS04,CGH10,MR99,Rei08} to name a few, and one excellent
example in the area of satisfiability is the derandomization of
Sch\"oning's algorithm for $k$-SAT.

In \cite{Sch99}, Sch\"oning proposed a simple randomized local search
algorithm for $k$-SAT, and showed that it runs in expected time
$\exprun((2-2/k)^n)$, which is $\exprun(1.3334^n)$ when $k=3$.
Later it was derandomized by Dantsin et al.~\cite{DGH02}.
They proposed a $k$-SAT algorithm that deterministically simulates
Sch\"oning's algorithm in time $\exprun((2-2/(k+1))^n)$,  which is
$\exprun(1.5^n)$ when $k=3$. 
Sch\"oning's algorithm makes use of randomness in the following two parts:
(i) choosing initial assignments for local search uniformly at random,
and
(ii) performing random walks as the local search.
Dantsin et al.~\cite{DGH02} derandomized it (i) by constructing a set of
Hamming balls  (so-called \textit{covering codes}), which efficiently
covers the entire search space $\{0,1\}^n$, and (ii) by replacing each
random walk by backtracking search.
Here (i) is ``perfectly'' derandomized in some sense, however, the
derandomization of (ii) loses some efficiency.
For 3-SAT, the efficiency in derandomizing part (ii) was gradually
improved by a sequence of works~\cite{DGH02,BK04,Sch08,KS10}.
Finally, and very recently, Moser and Scheder \cite{MS10} showed a 
{\em full} derandomization of Sch\"oning's algorithm, that is, they
proposed a deterministic algorithm for $k$-SAT that runs in time
$\exprun((2-2/k+\epsilon)^n)$ for any $\epsilon>0$.
The running time matches that of Sch\"oning's algorithm, and we now have
a deterministic $\exprun(1.3334^n)$ time algorithm for 3-SAT.

\subsection*{Our Contribution}

We investigate the possibility of derandomizing faster randomized
algorithms for 3-SAT.
In~\cite{HSS02}, Hofmeister, Sch\"oning, Schuler and Watanabe improved
Sch\"oning's algorithm for the 3-SAT case, 
that is, they proposed a randomized algorithm for 3-SAT that runs in
expected time $\exprun(1.3303^n)$.
Their improvement is based on a sophisticated way of randomly choosing
initial assignments rather than just choosing the ones uniformly at
random.

In this paper, we present a full derandomization of their algorithm,
that immediately implies the following result: 
\begin{theorem}\label{theo:main}
Problem 3-SAT is deterministically solvable in time $\exprun(1.3303^n)$.
\end{theorem}
As long as the authors know, it is the currently fastest deterministic
algorithm for 3-SAT.
Our result seems to be a necessary step towards a full derandomization
of the currently best known randomized algorithm, since it is based on
the combination of two algorithms~\cite{ISTT10} and~\cite{HMS11}, 
which are respectively a modification of Hofmeister~et~al.'s
algorithm~\cite{HSS02} and an extension of Paturi~et~al.'s
algorithm~\cite{PPS98}. 

To prove the above result, we develop a new way of explicitly
constructing covering codes with the properties which corresponds to the
distribution used to generate initial assignments in Hofmeister~et~al.'s
algorithm.

More precisely, we respectively denote by SCH and HSSW the randomized
algorithms by Sch\"oning \cite{Sch99}, and by Hofmeister, Sch\"oning,
Schuler, and Watanabe \cite{HSS02}.
Algorithm HSSW is obtained by modifying SCH, where one of the main
differences between SCH and HSSW is to choose initial assignments for
random walks as the local search: 
HSSW starts the random walk at an assignment chosen randomly from
$(\allbutfalse)^{\hat{m}}$ for some $\hat{m}\leq n/3$, while SCH starts
it at an assignment chosen uniformly from the whole space $\{0,1\}^n$.

We derandomized this random selection of initial assignments for HSSW
in the similar way to  SCH \cite{DGH02}, i.e., by constructing a
covering code (i.e., a set of balls  that covers the whole search space
$(\allbutfalse)^{\hat{m}}$).
However, due to the difference of $(\allbutfalse)^{\hat{m}}$ and
$\{0,1\}^n$, we cannot directly apply a \textit{uniform} covering code
developed in \cite{DGH02}. 
To efficiently cover the space  $(\allbutfalse)^{\hat{m}}$, we
introduced a generalized covering code, an \textit{$[\ell]$-covering
code}, which is a sequence of codes $C(0),C(1), \dots,C(\ell)$ such that 
(i) $C(i)$ is a set of balls of radius $i$, and 
(ii) $\bigcup_{i=0}^\ell C(i)$ covers $(\allbutfalse)^{\hat{m}}$. 
We remark that the generalized covering code has non-uniform covering
radius while an ordinary covering code has uniform radius.

We first show the existence of {\em small} $[\ell]$-covering code
$(C(0),C(1), \dots,C(\ell))$, and then similarly to \cite{DGH02},
by using an approximation algorithm for the set cover problem,
we show a deterministic construction of an $[\ell]$-covering code
$\tilde{C}(0),\tilde{C}(1),\dots,\tilde{C}(\ell)$
such that $|\tilde{C}(i)|\approx|C(i)|$.

\nop{
Finally,
given such a covering code
$\tilde{C}(0),\tilde{C}(1),\dots,\tilde{C}(\ell)$,
for each $0\leq i\leq\ell$ and for each $z\in\tilde{C}(i)$,
we search for a satisfying assignment
within the ball of radius $i$ centered at $z$.
We apply to this search
the algorithm which was recently developed in \cite{MS10}.
}

We  remark that our technique of constructing certain types of
covering codes has a potential application, for example, it 
can be applied to the further extensions~\cite{BS03,Rol03} of HSSW.



\section{Preliminaries}

In this section,
we briefly review HSSW algorithm for 3-SAT
proposed in \cite{HSS02}.
In what follows,
we focus on 3-CNF formulas.
Let $\varphi$ be a 3-CNF formula over $X=\{x_1,\dots,x_n\}$.
We alternatively regard $\varphi$ as the set of clauses of $\varphi$.
Thus,
the {\em size} of $\varphi$,
which is the number of clauses of $\varphi$,
is denoted by $|\varphi|$.
For any sub-formula $\varphi'\subset\varphi$ (resp., any clause $C\in\varphi$),
we denote by $X(\varphi')$ (resp., $X(C)$)
the set of variables of $\varphi'$ (resp., $C$).

A clause set $\varphi'\subset\varphi$ is {\em independent}
if $C\cap C'=\emptyset$ for any pair of clauses $C,C'\in\varphi'$.
An independent clause set $\varphi'$ is {\em maximal}
if for any clause $C\in(\varphi\setminus\varphi')$
there exists a clause $C'\in\varphi'$ such that $C\cap C'\neq\emptyset$.
For any partial assignment $t$ to $X(\varphi)$,
we denote by $\varphi|_t$
a sub-formula obtained from $\varphi$ by fixing variables according to  $t$.
Given a 3-CNF formula $\varphi$,
 algorithm HSSW starts
with arbitrarily finding a maximal independent clause set of $\varphi$.

\begin{fact}
Let $\varphi$ be a 3-CNF formula.
Let $\varphi'\subset\varphi$ be a maximal independent clause set of $\varphi$.
Then,
for any assignment $t$ to $X(\varphi')$,
the formula $\varphi|_{t}$ is a 2-CNF formula.
\end{fact}

Before describing HSSW,
we briefly review SCH algorithm for $k$-SAT
proposed in \cite{Sch99}.
Algorithm SCH is a randomized algorithm
which repeats the following procedure exponentially (in $n$) many times:
choose a random assignment $t$,
and run a random walk starting at $t$ as follows:
for a current assignment $t'$,
if $\varphi$ is satisfied by $t'$,
then output $\YES$ and halt.
Otherwise,
choose an arbitrary clause $C$ unsatisfied by $t'$,
and then update $t'$ by flipping the assignment of a variable of $C$
chosen uniformly at random.
This random walk procedure denoted by $\texttt{SCH-RW}(\varphi,t)$ is also exploited in HSSW.
The success probability of $\texttt{SCH-RW}(\varphi,t)$
for  a satisfiable $\varphi$
was analyzed in \cite{Sch99}:
Let $\varphi$ be a 3-CNF formula that is satisfiable.
Let $t_0$ be an arbitrary satisfying assignment of $\varphi$.
Then,
for  any initial assignment $t$ with Hamming distance $d(t_0,t)=r$,
we have 
\begin{equation}\label{eq:sch}
 \Pr\{\texttt{SCH-RW}(\varphi,t)=\YES\}
 \geq \left(\frac{1}{2}\right)^r\cdot\frac{1}{\poly(n)}.
\end{equation}

Now,
we are ready to present HSSW.
Given a 3-CNF formula $\varphi$,
HSSW first obtains a maximal independent clause set $\varphi'\subset\varphi$.
Note here that
the formula $\varphi|_t$ for any  assignment to $X(\varphi')$
is a 2-CNF, and hence we can check in polynomial time whether $\varphi|_t$ is satisfiable. From this observation, 
when  $\varphi'$ is small,
we can significantly improve the whole running time,
that is, it only requires $\exprun(7^{|\varphi'|})$ time.
On the other hand,
when the size of $\varphi'$ is large,
we repeatedly apply the random walk procedure $\texttt{SCH-RW}$. 
In this case,
we can also reduce the running time
by smartly choosing
initial assignments from satisfiable assignments of $\varphi'$:
Recall that SCH uniformly chooses initial assignments from $\{0,1\}^n$, 
which utilizes no information on $\varphi$.
Intuitively, HSSW uses initial assignments for \texttt{SCH-RW} that are closer to any satisfiable assignment.
In fact we can prove  
that the larger the size of $\varphi'$ is,
the higher
the probability that the random walk starts at an assignment
closer to a satisfying assignment is.

Formally,  algorithm HSSW is described  in Fig.\ \ref{fig:hssw}.
\begin{figure}[h]
\begin{quote}
$\texttt{HSSW}(\varphi)$ ~
// $\varphi$: a 3-CNF formula over $X$ \\ \\
\htab Obtain a maximal independent clause set $\varphi'\subset\varphi$ \\ \\
\htab If $|\varphi'|\leq\alpha n$, then \\
\htab\htab \textbf{for each} $t\in\{0,1\}^{X(\varphi')}$
that satisfies $\varphi'$ \\
\htab\htab\htab Check the satisfiability of $\varphi|_t$
~ // $\varphi|_t$: a 2-CNF formula \\ \\
\htab If $|\varphi'|>\alpha n$, then \\
\htab\htab $c$ \textbf{times do} \\
\htab\htab\htab Run $t=\texttt{init-assign}(X,\varphi')$ \\
\htab\htab\htab Run $\texttt{SCH-RW}(\varphi,t)$ \\
\htab\htab Output \texttt{NO} \\ \\
$\texttt{init-assign}(X,\varphi')$ ~
// return an assignment $t\in\{0,1\}^X$ defined as follows \\ \\
\htab \textbf{for each} $C\in\varphi'$ \\
\htab\htab Assume $C=x_i\vee x_j\vee x_k$ \\
\htab\htab Choose a random assignment $t$ to $x=(x_i,x_j,x_k)$ \\
\htab\htab ~ following the probability distribution:
\[\begin{array}{l}
 \Pr\{x=(1,0,0)\}
 =\Pr\{x=(0,1,0)\}
 =\Pr\{x=(0,0,1)\}=p_1 \\
 \Pr\{x=(1,1,0)\}
 =\Pr\{x=(1,0,1)\}
 =\Pr\{x=(0,1,1)\}=p_2 \\
 \Pr\{x=(1,1,1)\}=p_3
\end{array}
\]
\htab \textbf{for each} $x\in X\setminus X(\varphi')$ \\
\htab\htab Choose a random assignment $t$ to $x\in\{0,1\}$
\end{quote}
\caption{Algorithm \texttt{HSSW}}\label{fig:hssw}
\end{figure}
The algorithm contains 5 parameters
$\alpha$, $c$, and triple $(p_1,p_2,p_3)$ with $3p_1+3p_2+p_3=1$.
These parameters are  set to minimize  the whole expected running time.

Consider algorithm HSSW in Fig.\ \ref{fig:hssw} when $|\varphi'|>\alpha n$.
Let $\texttt{HSSW-RW}(\varphi')$ be the procedure
that is repeated $c$ times.
Then,
by using the lower bound (\ref{eq:sch}),
and setting parameters $(p_1,p_2,p_3)$ suitably
(c.f., Lemma \ref{lemm:hssw} below),
we have:
for any satisfiable 3-CNF formula $\varphi$,
\begin{equation}\label{eq:hssw-success}
 \Pr_{t,\texttt{SCH-RW}}\{\texttt{HSSW-RW}(\varphi')=\YES\}
 \geq
 \left(\frac{3}{4}\right)^{n}\cdot\left(\frac{64}{63}\right)^{|\varphi'|}.
\end{equation}
The whole expected running time $\exprun(1.3303^n)$ is obtained
by setting $\alpha$ to satisfy  the following equation.
\[
 \left(
 \left(\frac{3}{4}\right)^{n}\cdot\left(\frac{64}{63}\right)^{\alpha n}
 \right)^{-1}
 = 7^{\alpha n}.
\]
The values of parameters $(p_1,p_2,p_3)$ are determined
according to 
the following lemma,
which will be used by our derandomization.

\begin{lemma}
[Hofmeister, Sch\"oning, Schuler, and Watanabe \cite{HSS02}]
\label{lemm:hssw}
Let $\varphi$ be a 3-CNF formula that is satisfiable,
and let $\varphi'\subset\varphi$ be
a maximal independent clause set of $\varphi$.
Let $t$ be a random (partial) assignment
obtained via $\texttt{init-assign}(X,\varphi')$
and restricted to $X(\varphi')$.
Then,
for any (partial) assignment $t_0\in\{0,1\}^{X(\varphi')}$
that satisfies $\varphi'$,
\begin{equation}\label{eq:hssw}
 \mathop{\E}_t
 \left[\left(\frac{1}{2}\right)^{d(t_0,t)}\right]
 =\left(\frac{3}{7}\right)^{|\varphi'|}.
\end{equation}
\end{lemma}

There are two types of randomness that are used in HSSW:
(1) the random assignment obtained via \texttt{init-assign},
and
(2) the random walk of \texttt{SCH-RW}.
Fortunately,
the latter type of randomness can be (fully) removed by the recent result.
(Compare it with the inequality (\ref{eq:sch}).)

\begin{theorem}
[Moser and Scheder \cite{MS10}]\label{theo:moser}
Let $\varphi$ be a 3-CNF formula that is satisfiable.
Let $t_0$ be an arbitrary satisfying assignment of $\varphi$.
Given an assignment $t$ such that $d(t_0,t)=r$
for a non-negative integer $r$.
Then,
the satisfying assignment $t_0$ can be found deterministically
in time $\exprun((2+\epsilon)^r)$ for any constant $\epsilon>0$.
\end{theorem}

In the next section,
we show that
the former type of randomness is also not necessary.
It is shown by using covering codes,
that is in the similar way to \cite{DGH02}.
But,
the covering code we make use of is different from ordinary ones.
For any positive integer $n$,
a {\em code} of length $n$ is a subset of $\{0,1\}^n$,
where each element of a code is called a {\em codeword}.
A code $C\subset\{0,1\}^n$ is called an {\em $r$-covering code}
if for every $x\in\{0,1\}^n$,
there exists a codeword $y\in C$ such that $d(x,y)\leq r$.
This is the definition of an ordinary covering code.
We define a generalization of covering codes in the following way:

\begin{definition}
\rm
Let $\ell$ be a non-negative integer.
A sequence $C(0),C(1),\dots,C(\ell)$ of codes
is a {\em $\{0,1,\dots,\ell\}$-covering code},
or simply an {\em $[\ell]$-covering code},
if for every $x\in\{0,1\}^n$,
there exists a codeword $y\in C(r)$ for some $r:0\leq r\leq\ell$
such that $d(x,y)\leq r$.
\end{definition}

For ordinary covering codes,
it is easy to show the existence of a ``good'' $r$-covering code.
Moreover,
it is known that
we can deterministically construct such an $r$-covering code.

\begin{lemma}
[Dantsin et al.\ \cite{DGH02}]\label{lemm:dantsin}
Let $d\geq 2$ be a divisor of $n\geq 1$,
and let $0<\rho<1/2$.
Then,
there is a polynomial $q_d(n)$
such that a covering code of length $n$,
radius at most $\rho n$,
and size at most $q_d(n)2^{(1-h(\rho))n}$,
can be deterministically constructed
in time $q_d(n)(2^{3n/d}+2^{(1-h(\rho))n})$.
\end{lemma}

\section{A derandomization of HSSW}

In this section,
we prove Theorem \ref{theo:main} by derandomizing HSSW.
We do that in the similar way to \cite{DGH02}.
Let $\varphi$ be a 3-CNF formula,
and $\varphi'$ be a maximal independent clause set of $\varphi$.
Let $|\varphi'|=\hat{m}$,
and we suppose $\hat{m}=\Omega(n)$.
As is explained in the Introduction,
we will use a {\em generalized} covering code:
an $[\ell]$-covering code.
First,
we show that
there exists an $[\ell]$-covering code
for $(\{0,1\}^3\setminus 0^3)^{\hat{m}}$
where each of its codes is of small size.

\begin{lemma}\label{lemm:main}
For $(\{0,1\}^3\setminus 0^3)^{\hat{m}}$,
there exists an $[\ell]$-covering code $C(0),C(1),\dots,$ $C(\ell)$,
where
$\ell$ is the maximum integer
such that $(3/7)^{\hat{m}}<(1/2)^{\ell-2}$,
and $|C(i)|=O(\hat{m}^2(7/3)^{\hat{m}}/2^i)$.
\end{lemma}
\begin{proof}
We show the existence of such an $[\ell]$-covering code
by a probabilistic argument,
as is the case of the existence of an ordinary covering code for $\{0,1\}^n$.
However,
the probabilistic construction of an $[\ell]$-covering code
is different from the simple one of an ordinary covering code
in terms of,
(1) non-uniform covering radius,
and
(2) non-uniform choice of codewords.

For obtaining the desired covering code,
we make use of the probability distribution calculated in \cite{HSS02},
that is,
the equation (\ref{eq:hssw}) of Lemma \ref{lemm:hssw}.
The probabilistic construction is as follows:
Let $\ell$ be the integer defined above.
For each $i:0\leq i\leq\ell$,
let $C(i)\subset (\{0,1\}^3\setminus 0^3)^{\hat{m}}$ be a random code
obtained by choosing $y\in (\{0,1\}^3\setminus 0^3)^{\hat{m}}$
according to the distribution defined by the function $\texttt{init-assign}$
(in Fig.\ \ref{fig:hssw}),
and by repeating it independently $8\hat{m}^2(7/3)^{\hat{m}}/2^i$ times.

We will show that $C(0),C(1),\dots,C(\ell)$
is an $[\ell]$-covering code with high probability.
Fix $x\in(\{0,1\}^3\setminus 0^3)^{\hat{m}}$ arbitrarily.
Note here that
$\ell\leq 2\hat{m}$ and $(1/2)^{\ell-1}\leq(3/7)^{\hat{m}}$.
Then,
\begin{eqnarray*}
 & & \sum_{i=0}^{3\hat{m}} (1/2)^i\Pr_y\{d(x,y)=i\} \\
 &=&
 \sum_{i=0}^{\ell} (1/2)^i\Pr_y\{d(x,y)=i\}
 + \sum_{i=\ell+1}^{3\hat{m}} (1/2)^i\Pr_y\{d(x,y)=i\} \\
 &\leq&
 \sum_{i=0}^{\ell} (1/2)^i\Pr_y\{d(x,y)=i\} + (1/2)^{\ell} \\
 &\leq&
 \sum_{i=0}^{\ell} (1/2)^i\Pr_y\{d(x,y)=i\} + (3/7)^{\hat{m}}/2.
\end{eqnarray*}
Recall from the equation (\ref{eq:hssw}) of Lemma \ref{lemm:hssw} that,
\[
 \mathop{\E}_{y}\left[\left(\frac{1}{2}\right)^{d(x,y)}\right]
 = \sum_{i=0}^{3\hat{m}} (1/2)^i\Pr_y\{d(x,y)=i\}
 = \left(\frac{3}{7}\right)^{\hat{m}}.
\]
From these two,
we have
\[
 \sum_{i=0}^{\ell} (1/2)^i\Pr_y\{d(x,y)=i\}
 \geq (3/7)^{\hat{m}}/2.
\]
From this,
we see there exists an $r:0\leq r\leq \ell$
such that
\begin{equation}\label{eq:lower-bound}
 \Pr_y\{d(x,y)=r\} \geq (3/7)^{\hat{m}} 2^{r-1}/\ell.
\end{equation}
Note that
this value of $r$ depends on $x$.
Thus,
for each $x\in(\{0,1\}^3\setminus 0^3)^{\hat{m}}$,
if we define
\[
 r_x
 \mathdef
 \arg\max_{i:0\leq i\leq\ell}\left\{(1/2)^i\Pr\{d(x,y)=i\}\right\},
\]
we see that $r=r_x$ satisfies the above inequality (\ref{eq:lower-bound})
\footnote{
This definition of $r_x$ is not meaningful
if we merely show the existence.
However,
it is used when we consider a deterministic construction.
See the next lemma.
}.
Let $B(z,i)$ be the set of $w\in\{0,1\}^{3\hat{m}}$
such that $d(z,w)\leq i$.
Then,
from the lower bound (\ref{eq:lower-bound}),
the probability that $x$ is not covered with any $C(i)$ is
\begin{eqnarray*}
 \Pr_C
 \left\{x\not\in\bigcup_{i=0}^{\ell}\bigcup_{z\in C(i)}B(z,i)\right\}
 &\leq&
 \Pr_{C(r_x)}
 \left\{x\not\in\bigcup_{z\in C(r_x)}B(z,r_x)\right\} \\
 &=&
 \Pr_{C(r_x)}
 \left\{\forall y\in C(r_x)[d(x,y)>r_x]\right\} \\
 &=&
 \left(\Pr_y\left\{[d(x,y)>r_x]\right\}\right)^{|C(r_x)|} \\
 &=&
 \left(1-\Pr_y\{d(x,y)\leq r_x\}\right)^{|C(r_x)|} \\
 &\leq&
 \left(1-\Pr_y\{d(x,y)=r_x\}\right)^{|C(r_x)|} \\
 &\leq&
 \left(1-(3/7)^{\hat{m}} 2^{r_x-1}/\ell\right)^{|C(r_x)|} \\
 &\leq&
 \exp\left(-(3/7)^{\hat{m}} 2^{r_x-1}|C(r_x)|/\ell\right) \\
 &\leq&
 \exp\left(-2\hat{m}\right)
\end{eqnarray*}
Thus,
from the union bound,
the probability that some $x\in(\{0,1\}^3\setminus 0^3)^{\hat{m}}$
is not covered with any $C(i)$
is at most $7^{\hat{m}}\cdot\exp(-2\hat{m})=o(1)$.
Therefore,
there does exist an $[\ell]$-covering code stated in this lemma.
\qed
\end{proof}

Note that
this lemma only shows the {\em existence} of such an $[\ell]$-covering code.
We need to {\em deterministically} construct it.
However,
we can get around this issue in the same way as \cite{DGH02}:
applying the approximation algorithm for the set cover problem.
But,
since an $[\ell]$-covering code is not of uniform radius,
we can not directly apply the approximation algorithm.

\begin{lemma}\label{lemm:construction}
Let $d\geq 2$ be a constant that divides $\hat{m}$,
and let $\hat{m}'=\hat{m}/d$.
Let $\ell'$ be the maximum integer such that $(3/7)^{\hat{m}'}<(1/2)^{\ell'-2}$,
and let $s'_i=8\hat{m}'^2(7/3)^{\hat{m}'}/2^i$ for each $i:0\leq i\leq \ell'$.
Let $\ell=\ell' d$.
Then,
there is a polynomial $q_d(\hat{m})$ that satisfies the following:
an $[\ell]$-covering code $C(0),C(1),\dots,C(\ell)$
for $(\{0,1\}^3\setminus 0^3)^{\hat{m}}$
such that $|C(i)|\leq q_d(\hat{m})\cdot (7/3)^{\hat{m}}/2^i$
for $0\leq i\leq \ell$,
can be deterministically constructed
in time
$\poly(\hat{m})\cdot 7^{3\hat{m}/d}+q_d(\hat{m})\cdot (7/3)^{\hat{m}}$.
\end{lemma}
\begin{proof}
First,
we deterministically construct
an $[\ell']$-covering code $D'(0),D'(1),$ $\dots,D'(\ell')$
for $(\{0,1\}^3\setminus 0^3)^{\hat{m}'}$
such that $|D'(i)|\leq\poly(\hat{m}')\cdot s'_i$.
(Then, we concatenate all of them.
See below for details.)
Recall the proof of the previous lemma:
Let $p_i=(1/2)^i\Pr\{d(x,y)=i\}$ for each $i:0\leq i\leq\ell'$.
For any $x\in(\{0,1\}^3\setminus 0^3)^{\hat{m}'}$,
we have defined $r_x=\arg\max\{p_i:0\leq i\leq\ell'\}$,
which depends only on $x$.
Then,
we have concluded that
the sequence $C'(0),C'(1),\dots,C'(\ell')$ of random codes
satisfies the following with high probability:
every $x\in(\{0,1\}^3\setminus 0^3)^{\hat{m}'}$
is covered with the random code $C'(r_x)$.

Fix such an $[\ell']$-covering code $C'(0),C'(1),\dots,C'(\ell')$
for $(\{0,1\}^3\setminus 0^3)^{\hat{m}'}$.
Note here that $|C'(i)|=s'_i$ for $0\leq i\leq \ell'$.
For each $i:0\leq i\leq \ell'$,
let
\[
 A_i \mathdef \left\{x\in(\{0,1\}^3\setminus 0^3)^{\hat{m}'}:r_x=i\right\}.
\]
Note that
$[A_0,A_1,\dots,A_{\ell'}]$
is a partition of $(\{0,1\}^3\setminus 0^3)^{\hat{m}'}$.
Below,
we regard that
$C'(i)$ dedicates to covering (only) $A_i$
(although some codeword of $C'(i)$ may cover some elements outside $A_i$).

The point of the proof is that
we apply the approximation algorithm for the set cover problem to $A_i$
(not to the whole space $(\allbutfalse)^{\hat{m}'}$),
from which we (deterministically) obtain a covering code for $A_i$.
For this,
we obtain all elements of $A_i$ and keep them.
This is done by
calculating the value of $r_x$ for each $x\in(\allbutfalse)^{\hat{m}'}$,
Furthermore,
the calculation of $r_x$ is done
by calculating $p_j$ for every $j:0\leq j\leq\ell'$:
enumerate all $y\in(\{0,1\}^3\setminus 0^3)^{\hat{m}'}$ such that $d(x,y)=j$,
and then calculate the probability that $y$ is generated
by the function $\texttt{init-assign}$.
Then,
summing up those values of the probability,
we can calculate $\Pr\{d(x,y)=j\}$, and hence $p_j$.
Choosing $j$ as $r_x$
such that $p_j$ is the maximum of all $j:0\leq j\leq \ell'$,
we can obtain the value of $r_x$, and hence $A_i$.
In total,
it takes $\poly(\hat{m})\cdot 7^{2\hat{m}'}$ time for that procedure.

Now,
we apply the approximation algorithm for the set cover problem to each $A_i$.
As is similar to \cite{DGH02},
the approximation algorithm finds
a covering code $D'(i)$ for $A_i$ such that $|D'(i)|\leq q(\hat{m}')\cdot s'_i$
in time $q(\hat{m}')\cdot|A_i|^3$
for some polynomial $q(\hat{m}')$.
In total,
since $|A_i|\leq 7^{\hat{m}'}$,
it takes $q(\hat{m}')\cdot 7^{3\hat{m}'}$ time for that procedure.

So far,
we have obtained an $[\ell']$-covering code $D'(0),D'(1),\dots,D'(\ell')$
for $(\{0,1\}^3\setminus 0^3)^{\hat{m}'}$
such that $|D'(i)|\leq\poly(\hat{m}')\cdot s'_i$.
For each $0\leq i\leq \ell=\ell' d$,
let
\[
 C(i) \mathdef
 \{D'(i_1)\times D'(i_2)\times\cdots\times D'(i_d):
 i=i_1+i_2+\cdots+i_d,~0\leq i_j\leq \ell'\}.
\]
It is easy to see that
$C(0),C(1),\dots,C(\ell)$ is an $[\ell]$-covering code
for $(\allbutfalse)^{\hat{m}}$.
We (naively) estimate the upper bound on $|C(i)|$.
Let $i_1,i_2,\cdots,i_d$ be integers such that
$i=i_1+i_2+\cdots+i_d$ and $0\leq i_j\leq \ell'$.
Then,
\begin{eqnarray*}
 & & |D'(i_1)\times D'(i_2)\times\cdots\times D'(i_d)| \\
 &=& (\poly(\hat{m}'))^d\cdot
 \frac{8\hat{m}'^2(7/3)^{\hat{m}'}}{2^{i_1}}\cdot
 \frac{8\hat{m}'^2(7/3)^{\hat{m}'}}{2^{i_2}}\cdot\cdots\cdot
 \frac{8\hat{m}'^2(7/3)^{\hat{m}'}}{2^{i_d}} \\
 &=&
 (\poly(\hat{m}'))^d\cdot
 \frac{(8\hat{m}'^2)^d(7/3)^{\hat{m}'d}}{2^{i_1+\cdots+i_d}} \\
 &=&
 (\poly(\hat{m}'))^d\cdot\frac{(7/3)^{\hat{m}}}{2^i}.
\end{eqnarray*}
Since the number of combinations $i_1,\dots,i_d$
such that $i=i_1+\cdots+i_d$ and $0\leq i_j\leq \ell'$
is at most $(\ell'+1)^d$,
we have
\[
 |C(i)| \leq
 (\ell'+1)^d\cdot\poly(\hat{m}')\cdot\frac{(7/3)^{\hat{m}}}{2^i}
 \leq q_d(\hat{m})
 \cdot\frac{(7/3)^{\hat{m}}}{2^i}
\]
for some polynomial $q_d(\hat{m})$.

Finally,
we check the running time needed to construct $C(i)$.
It takes $q(\hat{m})\cdot 7^{3\hat{m}/d}$ time
to construct the $[\ell']$-covering code $D'(0),D'(1),\dots,D'(\ell')$
for $(\{0,1\}^3\setminus 0^3)^{\hat{m}'}$.
Furthermore,
it takes $\sum_{i=0}^{\ell}|C(i)|$
to construct the $[\ell]$-covering code $C(0),$ $C(1),\dots,C(\ell)$
for $(\{0,1\}^3\setminus 0^3)^{\hat{m}}$,
which is at most $q_d(\hat{m})\cdot (7/3)^{\hat{m}}$.
Summing up,
it takes
$q(\hat{m})\cdot 7^{3\hat{m}/d}+q_d(\hat{m})\cdot (7/3)^{\hat{m}}$ in total.
\qed
\end{proof}

Recall that $|\varphi'|=\hat{m}=\Omega(n)$.
Let $n'=n-3\hat{m}$,
which is the number of variables in $\varphi$ not appeared in $\varphi'$.
For the space $\{0,1\}^{n'}$,
we use an ordinary covering code,
that is guaranteed by Lemma \ref{lemm:dantsin}
to be deterministically constructed.

\begin{corollary}\label{coro:main}
Let $d$ be a sufficiently large positive constant,
and let $0<\rho<1/2$.
Then,
there is a polynomial $q_d(n)$ that satisfies the following:
an $\{i+\rho n':0\leq i\leq \ell\}$-covering code
$C(0+\rho n'),C(1+\rho n'),C(2+\rho n'),\dots,C(\ell+\rho n')$
for $(\allbutfalse)^{\hat{m}}\times\{0,1\}^{n'}$
such that
$|C(i)|\leq q_d(n)(7/3)^{\hat{m}}2^{(1-h(\rho))n'}/2^{i}$,
can be deterministically constructed
in time
$q_d(n)(7/3)^{\hat{m}}2^{(1-h(\rho))n'}$.
\end{corollary}
\begin{proof}
It is derived from the previous lemma
and Lemma \ref{lemm:dantsin}.
Given an $[\ell]$-covering code $C_1(0),C_1(1),\dots,C_1(\ell)$
for $(\allbutfalse)^{\hat{m}}$,
and a $\rho n'$-covering code $C_2(\rho n')$ for $\{0,1\}^{n'}$.
For each $0\leq i\leq \ell$,
let
\[
 C(i+\rho n') \mathdef C_1(i)\times C_2(\rho n').
\]
It is easy to see that
$C(0+\rho n'),C(1+\rho n'),C(2+\rho n'),\dots,C(\ell+\rho n')$
is an $\{i+\rho n':0\leq i\leq \ell\}$-covering code
for the space $(\allbutfalse)^{\hat{m}}\times\{0,1\}^{n'}$.
Furthermore,
$|C(i+\rho n')|\leq q_d(n)(7/3)^{\hat{m}}2^{(1-h(\rho))n'}/2^{i}$
for each $i:0\leq i\leq\ell$.
From the previous lemma,
if the constant $d$ is sufficiently large,
the running time for (deterministically)
constructing $C_1(0),C_1(1),\dots,C_1(\ell)$
is at most $q_d(\hat{m})(7/3)^{\hat{m}}$.
Similarly,
from Lemma \ref{lemm:dantsin},
the running time for (deterministically) constructing $C_2(\rho n')$
is at most $q_d(n')2^{(1-h(\rho))n'}$.
Thus,
the total running time is at most
\begin{eqnarray*}
 & & q_d(\hat{m})(7/3)^{\hat{m}} + q_d(n')2^{(1-h(\rho))n'}
 + \sum_{i=0}^{\ell}q_d(n)(7/3)^{\hat{m}}2^{(1-h(\rho))n'}/2^i \\
 &\leq&
 q_d(n)(7/3)^{\hat{m}}2^{(1-h(\rho))n'}
\end{eqnarray*}
for some polynomial $q_d(n)$.
\qed
\end{proof}

Now,
using this corollary,
we show a derandomization of HSSW,
and hence we prove Theorem \ref{theo:main}.
The outline of the deterministic algorithm is almost same as HSSW,
which is described in Fig. \ref{fig:hssw}.
We show the derandomization for the case of $|\varphi'|>\alpha n$.
Given $\varphi'$,
we deterministically construct
an $\{i+\rho n':0\leq i\leq \ell\}$-covering code
$C(0+\rho n'),C(1+\rho n'),C(2+\rho n'),\dots,C(\ell+\rho n')$,
as is specified in the proofs of
Lemma \ref{lemm:dantsin},
Lemma \ref{lemm:construction},
and Corollary \ref{coro:main}.
For any $z\in\{0,1\}^n$ and non-negative integer $i$,
we denote by $B(z,i)$ the set of $w\in\{0,1\}^n$ such that $d(z,w)\leq i$.
Then,
given such an $\{i+\rho n':0\leq i\leq \ell\}$-covering code,
we check whether there is a satisfying assignment
within $B(z,i+\rho n')$
for each $0\leq i\leq \ell$ and each $z\in C(i+\rho n')$.
It is easy to see that
this algorithm finds a satisfying assignment of $\varphi$
if and only if $\varphi$ is satisfiable.

We estimate the running time of the algorithm.
For any fixed $i$ and $z$,
the search of a satisfying assignment within $B(z,i+\rho n')$
is done in time $(2+\epsilon)^{i+\rho n'}$ for any small constant $\epsilon>0$,
which is guaranteed by Theorem \ref{theo:moser}.
Thus,
given an $\{i+\rho n':0\leq i\leq \ell\}$-covering code,
the running time for this task for all $B(z,i+\rho n')$ is at most
\begin{eqnarray*}
 & & q_d(n)\cdot\sum_{0\leq i\leq \ell}
 \left(\frac{(7/3)^{\hat{m}}}{2^i}\cdot 2^{(1-h(\rho))n'}\right)
 \cdot 2^{i+\rho n'}\cdot (1+\epsilon)^n \\
 &=&
 q_d(n)\cdot\left(\frac{7}{3}\right)^{\hat{m}}\cdot
 \left(2^{(1-h(\rho))n'}\cdot 2^{\rho n'}\right)\cdot (1+\epsilon)^n \\
 &=&
 q_d(n)\cdot\left(\frac{7}{3}\right)^{\hat{m}}\cdot
 \left(\frac{4}{3}\right)^{n'}\cdot (1+\epsilon)^n \\
 &=&
 q_d(n)\cdot
 \left(\frac{4}{3}\right)^{n}\cdot \left(\frac{63}{64}\right)^{\hat{m}}
 \cdot (1+\epsilon)^n,
 \quad (\because~ n'=n-3\hat{m})
\end{eqnarray*}
for some polynomial $q_d(n)$.
Note from the above corollary that
the running time for constructing $\{i+\rho n':0\leq i\leq \ell\}$-covering code
is less than the above value.
Thus,
the total running time in case of $|\varphi'|>\alpha n$
is at most $\exprun((4/3)^{n}(63/64)^{\hat{m}}(1+\epsilon)^n)$
for any $\epsilon>0$.
(Compare this value with the success probability of (\ref{eq:hssw-success}).)
On the other hand,
it is easy to see that
the running time in case of $|\varphi'|\leq\alpha n$
is at most $\exprun(7^{\hat{m}})$.
Therefore,
by setting $\alpha$
so that $(4/3)^{n}(63/64)^{\alpha n}(1+\epsilon)^n=7^{\alpha n}$ holds
(with $\epsilon>0$ arbitrarily small),
we obtain the running time $\exprun(1.3303^n)$.

\section{Conclusion}

We have shown a full derandomization of HSSW, and thereby present a
currently fastest deterministic algorithm for 3-SAT.
An obvious future work is to obtain a full derandomization of the
currently best known randomized algorithm for 3-SAT~\cite{HMS11}.
To do so, it seems to be required to derandomize Paturi~et~al.'s
algorithm~\cite{PPS98} completely.
Another possible future work is to extend HSSW algorithm to the
$k$-SAT case.
It leads to the fastest deterministic algorithms for $k$-SAT, combined
with the derandomization techniques of this paper and Moser and
Scheder~\cite{MS10}.

\end{document}